\documentclass[10pt]{article}

\usepackage{graphicx, amssymb, latexsym, amsfonts, amsmath, lscape, amscd,
amsthm, color, epsfig, mathrsfs, tikz, enumerate}

\newtheorem{theorem}{Theorem}[section]

\newtheorem{corollary}[theorem]{Corollary}

\newtheorem{lemma}[theorem]{Lemma}
\newtheorem{proposition}[theorem]{Proposition}

\newtheorem{observation}[theorem]{Observation}

\newcommand\DELETE[1]{}

\begin{document}

\title{{\bf On homomorphism of oriented graphs with respect to push operation}}
\author{ {\sc Sagnik Sen}\\
\mbox{}\\
{\small Indian Statistical Institute, Kolkata, India}
}

\date{\today}

\maketitle

\begin{abstract}
An oriented graph is a directed graph without any cycle of length at most 2. 
 To push a vertex of a directed graph is to reverse the orientation of the arcs incident to that vertex.
Klostermeyer and MacGillivray  defined push graphs which are equivalence class of oriented graphs with respect to vertex pushing operation. 
They studied the homomorphism of the equivalence classes of oriented graphs with respect to push operation.  
In this article, we further study the same topic and answer some of the questions asked in the above mentioned work. 
The anti-twinned  graph of an oriented graph is obtained by adding and pushing a copy of each of its vertices. 
In particular, we show that two oriented graphs are in a push relation if and only if they have isomorphic anti-twinned graphs. 
Moreover, we study oriented homomorphisms of outerplanar graphs with girth at least five, planar graphs and planar graphs with girth at least eight with respect to the push operation.  
\end{abstract}



\section{Introduction and preliminaries}

An {\textit{oriented graph}} 
 is a directed graph with no cycle of length 1 or 2. By replacing each edge of a simple graph $G$ with an arc (ordered pair of vertices)
  we obtain an oriented graph $ \overrightarrow{G}$;  $ \overrightarrow{G}$ is  an \textit{orientation} of $G$ and $G$ is the \textit{underlying graph} of $ \overrightarrow{G}$.
  We denote by $V( \overrightarrow{G})$  and $A( \overrightarrow{G}$)   respectively  the set of vertices  and arcs  of $ \overrightarrow{G}$.   
For an arc $ \overrightarrow{uv}$ the vertex $u$ is an \textit{in-neighbor} of $v$ and $v$ is an \textit{out-neighbor} of $u$. 
 The set of all in-neighbors and the set of all out-neighbors of $v$ are denoted by $N_{\overrightarrow{G}}^-(v)$ 
  and $N_{\overrightarrow{G}}^+(v)$, respectively.

Let $ \overrightarrow{G}$ and $ \overrightarrow{H}$ be two oriented graphs. 
A homomorphism of $ \overrightarrow{G}$ to $ \overrightarrow{H}$ is a
 mapping $\phi: V(\overrightarrow{G}) \rightarrow V(\overrightarrow{H})$ which preserves the arcs, that is, $uv \in A(\overrightarrow{G})$ implies $\phi(u)\phi(v) \in A(\overrightarrow{H})$. 
 We write $ \overrightarrow{G} \rightarrow  \overrightarrow{H}$ whenever there exists a 
 homomorphism of $ \overrightarrow{G}$ to $ \overrightarrow{H}$ and say that $ \overrightarrow{H}$ \textit{bounds} $ \overrightarrow{G}$.   
 The \textit{oriented chromatic number} $\chi_o( \overrightarrow{G})$ of an oriented graph $\overrightarrow{G}$ is 
 then the minimum order  of an oriented graph $\overrightarrow{H}$ such that
$\overrightarrow{G}$
admits a homomorphism to $\overrightarrow{H}$~\cite{orientedchi}.

  To \textit{push} a vertex $v$ of a directed graph $\overrightarrow{G}$ is to change the 
  orientations  of all the arcs (that is, to replace the arc $ \overrightarrow{uv}$ by $\overrightarrow{vu}$) incident with $v$. 
  Vertex pushing  of directed graphs has been studied by several 
  researchers~\cite{fisher-push, klostermeyer1, klostermeyer2, Mosesian, Pretzel1, Pretzel2, Pretzel3, Garyandwood}
  while Ochem and Pinlou~\cite{OPgirth4} used 
the push operation on oriented graphs for 
proving  the 
upper bounds of the  oriented chromatic number for  the families of triangle-free planar graphs and of 2-outerplanar graphs. 
Finally, Klostermeyer and MacGillivray brought these two popular field of studies 
together in their work~\cite{push} and considered the push operation on oriented graphs to define  equivalence classes of oriented graphs and studied homomorphisms between them.

  Two  oriented graphs ${\overrightarrow{G}}^{\text{\tiny{1}}}$
 and ${\overrightarrow{G}}^{\text{\tiny{2}}}$ are in a \textit{push relation} if it is possible to obtain ${\overrightarrow{G}}^{\text{\tiny{2}}}$ by pushing some vertices 
 of ${\overrightarrow{G}}^{\text{\tiny{1}}}$.
 Note that this push relation is in fact an
equivalence relation. A \textit{push graph}  $[\overrightarrow{G}]$ is an equivalance class of oriented
graphs 
(where ${\overrightarrow{G}}^{\text{\tiny{1}}}$ is an element of the equivalence class) with respect to the
above mentioned relation. 
An element ${\overrightarrow{G}}^{\text{\tiny{1}}}$  of the equivalence class $[\overrightarrow{G}]$ is a
\textit{presentation} of $[\overrightarrow{G}]$. 
We use the notation  ${\overrightarrow{G}}^{\text{\tiny{1}}} \in [\overrightarrow{G}]$ 
for
${\overrightarrow{G}}^{\text{\tiny{1}}}$  is a
presentation of $[\overrightarrow{G}]$.

Note that the graphs having a push relation have the same underlying graph. Hence, we can define the \textit{underlying graph} of a push graph $[\overrightarrow{G}]$ by the underlying graph of any presentation of it and 
denote it by $G$. 
The \textit{order} of a push graph is the number of vertices of its underlying graph, hence can be denoted by 
$|V( \overrightarrow{G})|$ or $|V(G)|$. 
Intuitively, we can treat a push graph as an oriented graph whose arcs, incedent to a vertex, are able to switch directions. 
Given  any oriented graph $ \overrightarrow{G}$ we can consider the corresponding push graph 
$ [\overrightarrow{G}]$.

A push graph $[\overrightarrow{G}]$ admits a homomorphism $\phi$ to an oriented graph 
$\overrightarrow{H}$ if there exist  a presentation 
${\overrightarrow{G}}^{\text{\tiny{1}} } \in [\overrightarrow{G}]$ such that $\phi$ is a homomorphism of ${\overrightarrow{G}}^{\text{\tiny{1}}}$  to 
$\overrightarrow{H}$. We write $ [\overrightarrow{G}] \rightarrow  \overrightarrow{H}$ whenever there exists a 
 homomorphism of $ [\overrightarrow{G}]$ to $ \overrightarrow{H}$.

A push graph $[\overrightarrow{G}]$ admits a homomorphism $\phi$ to a push graph 
$[\overrightarrow{H}]$ if there exist presentations 
${\overrightarrow{G}}^{\text{\tiny{1}} } \in [\overrightarrow{G}]$ and 
${\overrightarrow{H}}^{\text{\tiny{1}}} \in [\overrightarrow{H}]$   
such that $\phi$ is a homomorphism of ${\overrightarrow{G}}^{\text{\tiny{1}}}$  to 
${\overrightarrow{H}}^{\text{\tiny{1}}}$. 
 We write $ [\overrightarrow{G}] \rightarrow  [\overrightarrow{H}]$ whenever there exists a 
 homomorphism of $ [\overrightarrow{G}]$ to $ [\overrightarrow{H}]$ and say that $ [\overrightarrow{H}]$ \textit{bounds} $ [\overrightarrow{G}]$. 
      As in  general graph homomorphisms, for push graphs also, a bijective homomorphism whose inverse is also a 
     homomorphism is 
     an \textit{isomorphism}.

The push chromatic number $\chi_{p}( [\overrightarrow{G}])$ of the push graph $[\overrightarrow{G}]$ is the minimum order  of a push graph $[\overrightarrow{H}]$ such that
$[\overrightarrow{G}]$
admits a homomorphism to $[\overrightarrow{H}]$.  
It is equivalent to say that the push chromatic number $\chi_{p}( [\overrightarrow{G}])$ of the push graph $[\overrightarrow{G}]$ is the minimum order  of an oriented graph $\overrightarrow{H}$ such that
$[\overrightarrow{G}]$
admits a homomorphism to $\overrightarrow{H}$.

The push chromatic number $\chi_{p}(G)$ of  an undirected graph 
$G$  is the maximum of the push chromatic numbers of all the push graphs with underlying graph  $G$. The push chromatic number $\chi_{p}(\mathcal{F})$ of  a family $ \mathcal{F}$ of graphs is the maximum of the push  chromatic numbers of the graphs from the family $\mathcal{F}$.

\begin{figure}

\centering
\begin{tikzpicture}

\filldraw [black] (.5,1.5) circle (2pt) {node[below]{$v_2$}};
\filldraw [black] (.5,2.5) circle (2pt) {node[above]{$v_1$}};

\filldraw [black] (3.5,1.5) circle (2pt) {node[below]{$v_2^\prime$}};
\filldraw [black] (3.5,2.5) circle (2pt) {node[above]{$v_1^\prime$}};

\draw[->] (.5,1.5) -- (.5,2);
\draw[-] (.5,2) -- (.5,2.5);

\draw[->] (3.5,1.5) -- (3.5,2);
\draw[-] (3.5,2) -- (3.5,2.5);

\draw[-<] (.5,1.5) -- (2.5,2.16);
\draw[-] (2.5,2.16) -- (3.5,2.5);

\draw[-<] (3.5,1.5) -- (1.5,2.16);
\draw[-] (1.5,2.16) -- (0.5,2.5);

\draw[thick] (.5,2) ellipse (0.6cm and 1.1cm);
\draw[thick] (3.5,2) ellipse (0.6cm and 1.1cm);

\node at (-1,2) {$R(\overrightarrow{G})=$};

\node at (.5,3.8) {$V(\overrightarrow{G})$};

\node at (.5,3.4) {$||$};

\node at (.5,3.4) {$||$};

\node at (2,0) {\textbf{(a)}};


\filldraw [black] (7,1.5) circle (2pt) {node[below]{}};
\filldraw [black] (9,1.5) circle (2pt) {node[above]{}};

\filldraw [black] (7,3.5) circle (2pt) {node[below]{}};
\filldraw [black] (9,3.5) circle (2pt) {node[above]{}};

\draw[-<] (7,1.5) -- (8,1.5);
\draw[-] (8,1.5) -- (9,1.5);

\draw[-<] (7,3.5) -- (8,3.5);
\draw[-] (8,3.5) -- (9,3.5);

\draw[-<] (7,1.5) -- (7,2.5);
\draw[-] (7,2.5) -- (7,3.5);

\draw[->] (9,1.5) -- (9,2.5);
\draw[-] (9,2.5) -- (9,3.5);

\node at (8,0) {\textbf{(b)}};

\end{tikzpicture}

\caption{\textbf{(a)} The anti-twinned graph $R(\vec{G})$ of $[\vec{G}]$. \textbf{(b)} A push invarient graph $\vec{UC}_4$.}\label{fig anti-twinned graph orientable}

\end{figure}
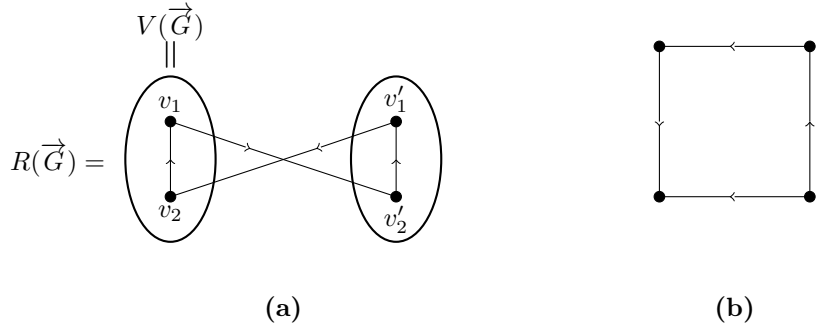

Klosetermeyer and MacGillivray~\cite{push} showed that deciding whether a push graph admits a $k$-push coloring or not is 
 NP-complete   for $k \geq 5$.  
In their paper, they  suggested several future directions regarding the topic among which we address here the following ones:

\begin{itemize}
\item[\textbf{(A)}] Is it true that two oriented graphs belong to  the same equivalence class  if and only if their anti-twinned (defined in Section~\ref{sec results}) graphs are isomorphic
 (see Fig.~\ref{fig anti-twinned graph orientable}(a))?

\item[\textbf{(B)}] What are the outerplanar graphs that have push chromatic number three.

\item[\textbf{(C)}] Investigating push chromatic  number for different graph families, especially, families of planar graphs.
\end{itemize}

In Section~\ref{sec results}  we address the above three points. First we will show that the question asked in \textbf{(A)} has a positive answer.  
Then we prove that outerplanar with \textit{girth} (length of the smallest cycle) at least five admits a push 3-coloring 
while observing that it is not possible to relax the girth restriction from that result. This will partially answer point \textbf{(B)}.
Then we deal with point \textbf{(C)} and  show that the push chromatic number of planar graphs lies between 10 and 40. 
Moreover, we prove that the push chromatic number for the family of planar graphs with girth 8 is 4.

\section{Results}\label{sec results}
The anti-twinned graph $R(\overrightarrow{G})$ of an oriented graph 
  $\overrightarrow{G}$
  was defined and used by Klostermeyer and MacGillivray in~\cite{push}.

Let $[\overrightarrow{G}]$ be a push graph with vertex set $V(G) = \{v_1, v_2, ..., v_k \}$ and 
$ {\overrightarrow{G}}^{\text{\tiny{1}}} \in [\overrightarrow{G}]$. Then the
\textit{anti-twinned graph} $R(\overrightarrow{G})$ of $[\overrightarrow{G}]$ is the oriented graph with the set of vertices and the set of arcs as the following (also see Fig.~\ref{fig anti-twinned graph orientable}(a)):
 
 \begin{align}\nonumber
 V(R(\overrightarrow{G})) = \{v_1, v_2, ..., v_k \} \cup \{v_1^\prime, v_2^\prime, ..., v_k^\prime \} \\ \nonumber
 A(R(\overrightarrow{G})) = \{\overrightarrow{v_iv_j}, \overrightarrow{v_i^\prime v_j^\prime}, \overrightarrow{v_jv_i^\prime}, \overrightarrow{v_j'v_i} \mid \overrightarrow{v_iv_j} \in A({\overrightarrow{G}}^{\text{\tiny{1}}}) \}. \nonumber
 \end{align}

  Intuitively, $R(\overrightarrow{G})$ is the graph obtained from $[\overrightarrow{G}]$ by adding and pushing  a twin vertex $v_i^\prime$  for each of the vertices $v_i$ of ${\overrightarrow{G}}^{\text{\tiny{1}}}$. 
  Observe that $R(\overrightarrow{G})$ is well defined upto isomorphism, that is, for any presentation of 
  $[\overrightarrow{G}]$, we will get the same oriented graph $R(\overrightarrow{G})$.

Now we will prove a result which answers point \textbf{(A)}
mentioned in the introduction.

 \begin{theorem}\label{corollary push and oriented}
Two oriented graphs $\overrightarrow{G}$ and  $\overrightarrow{H}$ are in the same  equivalence class with respect to the push operation if and only if 
their corresponding anti-twinned graphs $R(\overrightarrow{G})$ and $R(\overrightarrow{H})$ are isomorphic.
 \end{theorem}

 \begin{proof}If $\overrightarrow{G}$ and  $\overrightarrow{H}$ are in the same push equivalence class then 
 their corresponding anti-twinned graphs $R(\overrightarrow{G})$ and $R(\overrightarrow{H})$ are isomorphic was shown by Klostermeyer and MacGillivray~\cite{push}. So it is enough to prove only the ``if'' part of the theorem.

For any isomorphism $h$ of $R(\overrightarrow{G})$ to $R(\overrightarrow{H})$ define the set 
$$Y_h = \{v \in V(\overrightarrow{G}) | h(v') = h(v)'\}.$$ 
Furthermore,  if $x$ is a vertex of an oriented graph $\overrightarrow{X}$ then $x'$ denotes its corresponding anti-twin in $R(\overrightarrow{X})$. Moreover, we fix the convention $x'' = x$.

 Let $\overrightarrow{G}$ and $\overrightarrow{H}$ be two oriented graphs and let $f$ be an isomorphism of 
 $R(\overrightarrow{G})$ to $R(\overrightarrow{H})$. Note that if $Y_f = V(\overrightarrow{G})$ then we are done.

 Therefore, let $v \in V(\overrightarrow{G}) \setminus Y_f$, that is,  $f(v') \neq f(v)'$. Now we define the following:

$$g(x) = \begin{cases} f(x) &\mbox{if } x \neq v', f^{-1}(f(v)'), \\
f(v') & \mbox{if } x = f^{-1}(f(v)'), \\
f(v)' & \mbox{if } x = v'.\end{cases} $$

Intuitively, we just interchanged the images of $v'$ and $f^{-1}(f(v)')$ to obtain $g$. As $f$ was a bijective function from $V(R(\overrightarrow{G}))$ to $V(R(\overrightarrow{H}))$, so is $g$. 
If we can show that both $g$ and $g^{-1}$ are oriented graph homomorphisms between $R(\overrightarrow{G})$ and $R(\overrightarrow{H})$, then we will 
end up proving that $g$ is an oriented graph isomorphism of 
 $R(\overrightarrow{G})$ to $R(\overrightarrow{H})$. For convenience let $f^{-1}(f(v)') = u$.

First we will show  that $g$ is a homomorphism of $R(\overrightarrow{G})$ to $R(\overrightarrow{H})$. Let $\overrightarrow{ab}$ be an arc in 
$R(\overrightarrow{G})$. If $a,b \notin \{v', u\}$ then there is an arc from $g(a)=f(a)$ to $g(b)=f(b)$ as $f$ itself  
is an isomorphism. 

Now suppose $a = v'$ and $b \notin \{v', u\}$. This implies $b \in N^+(v') = N^-(v)$. Therefore, 
$$g(b)= f(b) \in N^-(f(v)) = N^+(f(v)')
= N^+(g(a)).$$
Similarly, one can argue for the case when $b = v'$ and $a \notin \{v', u\}$.

Then suppose that $a = u$ and $b \notin \{v', u)\}$. 
This implies $b \in N^+(u)$. 
Hence $$g(b) = f(b) \in N^+((f(v)') = N^-((f(v)) = N^+(f(v')) = N^+(g(f^{-1}(f(v)')) = N^+(g(a)).$$ 
Similarly, one can argue for the case when $b = u$ and $a \notin \{v', u\}$.

Note that $f(v)$ and $f(v)'$ are non-adjacent. Hence, $v = f^-(f(v)$ and $f^{-1}(f(v)')$ are also non-adjacent. Therefore, $v'$ and $f^{-1}(f(v)')$
are non-adjacent. That takes care of the case when $\{a,b\} = \{v', u\}$.

\medskip

Now we will prove that $g^{-1}$ is a homomorphism. 
Let $\overrightarrow{ab}$ be an arc in 
$R(\overrightarrow{H})$.
If $a,b \notin \{f(v'), f(v)'\}$ then there is an arc from $g^{-1}(a)=f^{-1}(a)$ to $g(b)=f^{-1}(b)$ as $f^{-1}$ itself  
is an isomorphism.

Now suppose $a = f(v')$ and $b \notin \{f(v'), f(v)'\}$. 
This implies $$b \in N^+(a) = N^+(f(v')) = N^-(f(v)) = N^+(f(v)').$$ 
Therefore, 
$$g^{-1}(b)= f^{-1}(b) \in N^+(f^{-1}(f(v)')) = N^+(u) =N^+(g^{-1}(g(u)))=$$
$$= N^+(g^{-1}(f(v'))) = N^+(g^{-1}(a)).$$
Similarly, one can argue for the case when $b = f(v')$ and $a \notin \{f(v'), f(v)'\}$.

Then suppose that $a = f(v)'$ and $b \notin \{f(v'), f(v)'\}$. 
This implies $$b \in N^+(a) = N^+(f(v)') = N^-(f(v)) = N^+(f(v')).$$ 
Hence $$g^{-1}(b)= f^{-1}(b) \in N^+(f^{-1}(f(v'))) =N^+(v'))=$$
$$= N^+(g^{-1}(f(v)')) = N^+(g^{-1}(a)).$$
Similarly, one can argue for the case when $b = f(v')$ and $a \notin \{f(v'), f(v)'\}$. 

\medskip

So we have shown that $g$ is an isomorphism. Also note that $u, u' \notin Y_f$, $Y_f \subset Y_g$ and either $u \in Y_g$
or $u' \in Y_g$. 
Hence $Y_f \varsubsetneq Y_g$. 

So we can recursively define a chain of isomorphisms 
$g= g_0, g_1, ..., g_t$ of $R(\overrightarrow{G})$ to $R(\overrightarrow{H})$  so that we have 
$$Y_{g_0}  \varsubsetneq  Y_{g_1}  \varsubsetneq .... \varsubsetneq Y_{g_t}  = V(\overrightarrow{G}).$$
This completes the proof.
 \end{proof}

\bigskip

Now we will prove a result that will give us more insight regarding the relation between oriented graph homomorphism and push operation and also help us to prove a particular step of an upcoming theorem.

\begin{proposition}\label{th any target}
Let $f$ be a homomorphism of $ \overrightarrow{G}$ to $ \overrightarrow{H}$. 
Then for each presentation ${\overrightarrow{H}}^{\text{\tiny{1}}} \in [\overrightarrow{H}]$ there exists a presentation 
${\overrightarrow{G}}^{\text{\tiny{1}}} \in [\overrightarrow{G}]$ such that $f$ is a homomorphism of ${\overrightarrow{G}}^{\text{\tiny{1}}}$
to ${\overrightarrow{H}}^{\text{\tiny{1}}}$.
\end{proposition}

\begin{proof}
Let $f$ be a homomorphism of $ \overrightarrow{G}$ to $ \overrightarrow{H}$ and ${\overrightarrow{H}}^{\text{\tiny{1}}} \in [\overrightarrow{H}]$ 
be any presentation.  Suppose one can obtain ${\overrightarrow{H}}^{\text{\tiny{1}}}$ from $ \overrightarrow{H}$ by pushing the set of vertices 
$V_1 \subseteq V( \overrightarrow{H})$. Now obtain the presentation ${\overrightarrow{G}}^{\text{\tiny{1}}}$ from $\overrightarrow{G}$ by 
pushing the pre-images of $V_1$, that is, the set of vertices $f^{-1}(V_1) \subseteq V( \overrightarrow{G})$. It is easy to check that 
$f$ is a homomorphism of ${\overrightarrow{G}}^{\text{\tiny{1}}}$
to ${\overrightarrow{H}}^{\text{\tiny{1}}}$.
\end{proof}

A \textit{splitable oriented graph} $\overrightarrow{S}$ is an oriented graph isomorphic  to the anti-twinned graph  $R(\overrightarrow{T})$ of some oriented graph $\overrightarrow{T}$. The oriented graph $ \overrightarrow{T}$ is the \textit{split graph} of $ \overrightarrow{S}$.
The following two results will be instrumental in proving other results of this article.

\begin{observation}\label{push why splitable}
An oriented graph $ \overrightarrow{S}$ is splitable if and only if it is possible to partition the set of vertices  $V(\overrightarrow{S})$ into two equal parts $V_1$ and $V_2$ with a bijection $f: V_1 \rightarrow V_2$ such that 
 for all $u \in V_1$ we have $N^+(u) = N^-(f(u))$ and $N^-(u) = N^+(f(u))$. 
\end{observation}

The above result follows directly from the definition of splitable oriented graph. 

\begin{lemma}\label{lem split implies good}
Let $\overrightarrow{S} = R(\overrightarrow{T})$ be a splitable graph. Then  $\overrightarrow{G} \rightarrow \overrightarrow{S}$
if and only if $[\overrightarrow{G}] \rightarrow \overrightarrow{T}$. 
\end{lemma}

\begin{proof}
Let $\overrightarrow{S} = R(\overrightarrow{T})$ be a splitable graph. 
Assume that $\overrightarrow{G} \rightarrow \overrightarrow{S}$. This implies 
$[\overrightarrow{G}] \rightarrow [\overrightarrow{S}]$. 
Now consider the following function $\psi$ from $V(R(\overrightarrow{T}))$ to $V(\overrightarrow{T})$:

\begin{align}\nonumber
\psi(x) = \psi(x') = x \text{ for } x \in V(\overrightarrow{T}).
\end{align}

It is easy to check that $\psi$ is a homomorphism of 
$R(\overrightarrow{T})$ to $\overrightarrow{T}$. This implies that there exists a 
push homomorphism $[R(\overrightarrow{T})] \rightarrow [\overrightarrow{T}]$. By composing this 
homomorphism with the homomorphism $[\overrightarrow{G}] \rightarrow [\overrightarrow{S}]$
we obtain a homomorphism $[\overrightarrow{G}] \rightarrow \overrightarrow{T}$. 
This proofs the ``only if'' part.

For proving the  ``if'' part assume $[\overrightarrow{G}] \rightarrow \overrightarrow{T}$. Then 
we have $R(\overrightarrow{G}) \rightarrow \overrightarrow{S}$ due to Klostermeyer and MacGillivray~\cite{push}. 
By composing this homomorphism with the inclusion homomorphism of 
$\overrightarrow{G}$ to $R(\overrightarrow{G})$ we will be done.
\end{proof}

An outerplanar graph is a graph that can be drawn on a plane with all its vertices lying on a circle and all its edges can be drawn inside the circle without any crossing. Klostermeyer and MacGillivray~\cite{push} showed that the family of outerplanar graphs has 
push chromatic number 4.

\begin{theorem}\mbox{}\label{pushouter}
Let $\mathcal{O}_5$ be the family of outerplanar graph with girth at least five. Then  $\chi_{p} (\mathcal{O}_5) = 3$.
\end{theorem}

\begin{proof}
In~\cite{outergirth}, Pinlou and Sopena showed that every outerplanar graph with girth at least $k$
and minimum degree at least 2 contains a face of length $l \geq k$ with at least $(l-2)$ consecutive vertices of degree 2.
We will show that every push outerplanar graph of girth at least 5 admits a homomorphism to the directed 3-cycle $ \overrightarrow{C_3}$.

Let $ [\overrightarrow{M}]$ be a minimal (with respect to inclusion as a
subgraph) push outerplanar graph with girth at least 5 having no homomorphism to $ \overrightarrow{C_3}$.

\begin{itemize}

\item[(i)]  Suppose that $[\overrightarrow{M}]$ contains a vertex $u$ of degree 1. Then, due  to the minimality of 
$ [\overrightarrow{M}]$, 
the push outerplanar graph obtained by deleting the vertex $u$ from $ [\overrightarrow{M}]$ (which has girth at least 5) admits a homomorphism to $\overrightarrow{C_3}$. 
Since every vertex of $\overrightarrow{C_3}$ has in-degree and out-degree equal to 1, the homomorphism can easily be extended to obtain a homomorphism of $[\overrightarrow{M}]$ to  $\overrightarrow{C_3}$, 
a contradiction.

\item[(ii)] Suppose that $[\overrightarrow{M}]$ contains a face $ux_1x_2...x_{l-2}v$ of length $l \geq 5$ with at least $(l-2)$ consecutive vertices $x_1,x_2, ...,x_{l-2}$ of degree 2. Then, due  to the minimality of 
$ [\overrightarrow{M}]$, 
the push outerplanar graph $[\overrightarrow{M'}]$ obtained by deleting the vertices $x_1,x_2, ...,x_{l-2}$
 from $ [\overrightarrow{M}]$ (which has girth at least 5) admits a homomorphism $\phi$ to $\overrightarrow{C_3}$. 
 Now, let ${\overrightarrow{M'}}^{\text{\tiny{1}}}$ be a presentation of $[\overrightarrow{M'}]$ with 
 $\phi: {\overrightarrow{M'}}^{\text{\tiny{1}}} \rightarrow \overrightarrow{C_3}$. Note that 
 the vertices $u$ and $v$ are adjacent in ${\overrightarrow{M'}}^{\text{\tiny{1}}}$. 
 Hence, $\phi(u) \neq \phi(v)$. 

It is possible to check that (a bit tidious, but not difficult, by case analysis), given any oriented path of length $m \geq 4$, with edges $a_1a_2, a_2a_3, ...,a_{m-1}a_{m}$ and a mapping 
$\psi: \{a_1, a_m\} \rightarrow V(\overrightarrow{C_3})$ with $\psi(a_1) \neq \psi(a_m)$, it is possible to push the vertices $a_i$ for $i \in \{2,..,m-1\}$ to obtain an oriented path and extend the mapping $\psi$ to a homomorphism of that oriented path to $\overrightarrow{C_3}$.

Hence, by the above observation, we can extend the homomorphism of $[\overrightarrow{M'}]$ to $ \overrightarrow{C_3}$ to a homomorphism of $[\overrightarrow{M}]$ to $ \overrightarrow{C_3}$, 
a contradiction.

As any  cycle of odd length has push chromatic number at least 3, the  bound is tight. 
\end{itemize}
\end{proof}

\medskip

An undirected simple  graph $G$ admits an \textit{acyclic $k$-coloring} if it can be colored by $k$ colors 
in such a way that the graph induced by each color is an independent set and the graph induced by a pair of colors is a forest. 
This definition was introduced in~\cite{acyclic-def}.

\begin{theorem}\label{pushacyclic}
Every graph with acyclic chromatic number at most $k$ has push chromatic number at most $k.2^{k-2}$. 
\end{theorem}

\begin{proof}
For any positive integer $k$ the \textit{Zielonka graph}~\cite{orientedchi} $\overrightarrow{Z}_k$  of order $k \times 2^{k-1}$ is  the oriented graph with    set of vertices $V(\overrightarrow{Z}_k) = \cup_{i=1,2,...,k} S_i$  where 
\begin{align}\nonumber
    S_i = \lbrace x = (x^{1},...,x^{k}) \vert x^{j} \in \lbrace 0,1 \rbrace \text{ for } j \neq i 
    \text{ and } x^{i} = * \rbrace 
\end{align}

\noindent and  set of arcs  

\begin{align}\nonumber
 A(Z_k) &= \lbrace \overrightarrow{xy} \mid  x = (x^{1},...,x^{k}) \in S_i , y = (y^{1},...,y^{k}) \in S_j 
 \text{ and }   \\ \nonumber
 & \hspace{1cm} \text{ either } x^{j} = y^{i} \text{ and }  i < j \text{  or } x^{j} \neq y^{i} \text{ and }  i > j \rbrace.
 \\ \nonumber 
\end{align}
 
\medskip 
 
Furthermore, note that the vertices of a  $\overrightarrow{Z}_k$ can be partitioned into two disjoint sets of equal size

$$V_1 = \{x = (x^{1},...,x^{k}) \in S_i \vert x^{1} + x^2 + ... +x^{i-1} + x^{i+1} + ... +x^{k} \geq \lceil k/2 \rceil \},$$
$$V_2 = \{x = (x^{1},...,x^{k}) \in S_i \vert x^{1} + x^2 + ... +x^{i-1} + x^{i+1} + ... +x^{k} < \lceil k/2 \rceil \}.$$

Also, we can define a function $f : V_1 \rightarrow V_2$ by $$f((x^{1},...,x^{i-1}, *, x^{i+1}, ...,x^{k})) = (x^{1}+1,...,x^{i-1}+1, *, x^{i+1}+1, ..., x^{k}+1)$$ where the $+$ operation is taken modulo 2. It is clear that $f$ is a bijection of the type described in Observation~\ref{push why splitable}. 
Hence $\overrightarrow{Z_k}$ is 
a splitable oriented graph. 
Therefore, $\overrightarrow{Z_k} = R(\overrightarrow{Z_k[V_1]})$.

Now let $[\overrightarrow{G}]$ be a push graph and its underlying simple graph $G$ admits an acyclic $k$-coloring. 
Raspaud and Sopena~\cite{planar80} showed that the oriented graph $\overrightarrow{G}$ admits an oriented homomorphism to $Z_k$. 
Hence by Lemma~\ref{lem split implies good} $[\overrightarrow{G}] \rightarrow \overrightarrow{Z_k[V_1]}$ where $\overrightarrow{Z_k[V_1]}$ 
is the oriented graph induced from $\overrightarrow{Z_k}$ by $V_1$. 
Clearly, $\overrightarrow{Z_k[V_1]}$ is a graph on 
$k.2^{k-2}$ vertices, hence we are done. 
\end{proof}

As 
$\chi_{p}([\overrightarrow{G}]) \leq \chi_{o}( \overrightarrow{G}) \leq 2\chi_{p}( [\overrightarrow{G}])$~\cite{push}, 
the upper bound of the above result is tight for $k \geq 3$ due to Ochem~\cite{Ochem_negativeresults}.
Now we establish the lower and upper bounds for the push chromatic number of  planar graphs.

\begin{theorem}\label{th p3}
Let $\mathcal{P}_3$ be the family of all planar graphs. Then $10 \leq \chi_{p} ( \mathcal{P}_3) \leq 40$. 
\end{theorem}

\begin{proof}
Borodin~\cite{Borodinacyclic} 
showed that every planar graph 
admits an acyclic 5-coloring. 
Hence the upper bound follows by Theorem~\ref{pushacyclic}.

\medskip

Now we will prove the lower bound. 

\medskip

\textbf{\textit{Claim 1:}} There exists an oriented graph $ \overrightarrow{H}$ on $\chi_{p}(\mathcal{P}_3)$ vertices such that every
planar push graph admit a homomorphism to $ \overrightarrow{H}$. 

\medskip

\textit{Proof of the claim:} Let $OG_{\chi_{p}(\mathcal{P}_3)}$ be the set of all oriented graphs of order  
$\chi_{p}(\mathcal{P}_3)$. If our claim is false then for each $ \overrightarrow{G} \in OG_{\chi_{p}(\mathcal{P}_3)}$ there exists a planar push graph 
$[\overrightarrow{P}_G]$ that does not admit a homomorphism to $\overrightarrow{G}$. 
Let $ [\overrightarrow{X}] = \bigsqcup_{\overrightarrow{G} \in OG_{\chi_{p}(\mathcal{P}_3)}}[\overrightarrow{P}_G]$ be the disjoint union of all such graphs. Note that $ [\overrightarrow{X}]$ is a planar graph and does not admit homomorphism to any oriented graph on $\chi_{p}(\mathcal{P}_3)$ vertices and thus have $\chi_{p}([\overrightarrow{X}]) > \chi_{p}(\mathcal{P}_3)$, a contradiction. \hfill $\bullet$

\medskip

We say an oriented graph $\overrightarrow{T}$ has property \textbf{(P1)} if $ \overrightarrow{T}$ is an oriented graph on 
$\chi_{p}(\mathcal{P}_3)$   vertices such that every
planar push graph admits an homomorphism to $ \overrightarrow{T}$. 
Moreover,  $ \overrightarrow{T}$ is minimal such graph with respect to subgraph inclusion. 

By the above claim we are garunteed that there exists an oriented graph $\overrightarrow{T}$ with property \textbf{(P1)}.
First note that if $ \overrightarrow{T}$ has \textbf{(P1)} then by Proposition~\ref{th any target} any presentation of $[\overrightarrow{T}]$ also have  property \textbf{(P1)}.
Note that the graph obtained by reversing all the arcs of  any graph with  property \textbf{(P1)} also has property \textbf{(P1)}. 

Two vertices $x$ and $y$ of an oriented graph agree on a third vertex $z$ if $z \in N^{\alpha}(x) \cap  N^{\alpha}(y)$ for 
some $\alpha \in \{+,-\}$. Similarly, 
$x$ and $y$ disagree on  $z$ if $z \in N^{\alpha}(x) \cap  N^{\beta}(y)$ for $\{\alpha, \beta\} = \{+,- \}$. Let $A_{x,y}$ denote the set of vertices that $x$ and $y$ agree on and let $D_{x,y}$ denote the set of vertices that $x$ and $y$ disagree on. Note that 
given two vertices $x$ and $y$ of a fixed oriented graph the sets $A_{x,y}$ and $D_{x,y}$ remains as it is under push operation unless you push exactly one of $x,y$. If you push exactly one of $x, y$ then the two sets get interchanged. 
Therefore, the parameters $M_{x,y} = max\{A_{x,y}, D_{x,y}\}$ and $m_{x,y} = min\{A_{x,y}, D_{x,y}\}$ are push invarient.

Let $ \overrightarrow{ab}$ be an arc of $ \overrightarrow{T}$. Then push all in-neighbors of $a$. Now if in the so-obtained presentation ${\overrightarrow{T}}^{\text{\tiny{1}}} \in [\overrightarrow{H}]$ we have $|A_{a,b}| \geq  |D_{a,b}|$, then stop. 
Otherwise, 
push the vertex $x$ and reverse all the arcs of ${\overrightarrow{T}}^{\text{\tiny{1}}}$ and stop. Call the graph obtained in the final step 
$ \overrightarrow{H}$. Notice that,  $ \overrightarrow{H}$  has property \textbf{(P1)}. Moreover, all neighbors of $a$ in 
$ \overrightarrow{H}$ are out-neighbors and we have $|A_{a,b}| \geq  |D_{a,b}|$.

\medskip

\textbf{\textit{Claim 2:}} If two vertices $x$ and $y$  of a push graph has $m_{x,y} \geq 1$ then they cannot have the same image under any homomorphism.

\medskip

\textit{Proof of the claim:} Let the 4-cycle drawn in Fig~\ref{fig anti-twinned graph orientable}(b) be the graph $\overrightarrow{UC}_4$. Note that this graph is push invarient and 
no two of its vertices can  have the same homomorphic image.
Also any two vertices of a push graph with $m_{x,y} \geq 1$ must be part of a subgraph isomorphic to $\overrightarrow{UC}_4$.
\hfill $\bullet$

\medskip

 \medskip

\begin{figure}

\centering
\begin{tikzpicture}


\filldraw [black] (8,0) circle (2pt) {node[below]{$x_1$}};
\filldraw [black] (10,0) circle (2pt) {node[below]{$x_2$}};
\filldraw [black] (12,0) circle (2pt) {node[below]{$x_3$}};

\filldraw [black] (10,2) circle (2pt) {node[right]{$x_4$}};

\filldraw [black] (8,4) circle (2pt) {node[above]{$x_5$}};
\filldraw [black] (10,4) circle (2pt) {node[above]{$x_6$}};
\filldraw [black] (12,4) circle (2pt) {node[above]{$x_7$}};

\filldraw [black] (6,2) circle (2pt) {node[left]{$x_8$}};


\draw[->] (8,4) -- (9,3);
\draw[-] (9,3) -- (10,2);

\draw[->] (10,4) -- (10,3);
\draw[-] (10,3) -- (10,2);

\draw[->] (12,4) -- (11,3);
\draw[-] (11,3) -- (10,2);

\draw[->] (10,2) -- (9,1);
\draw[-] (9,1) -- (8,0);

\draw[->] (10,2) -- (10,1);
\draw[-] (10,1) -- (10,0);

\draw[->] (10,2) -- (11,1);
\draw[-] (11,1) -- (12,0);

\draw[->] (8,0) -- (9,0);
\draw[-] (9,0) -- (10,0);

\draw[->] (10,0) -- (11,0);
\draw[-] (11,0) -- (12,0);

\draw[->] (8,4) -- (9,4);
\draw[-] (9,4) -- (10,4);

\draw[->] (10,4) -- (11,4);
\draw[-] (11,4) -- (12,4);


\draw (10,0) .. controls (9,-2) and (8,-2) .. (6,2);

\draw[->] (8.8,-1.3) -- (8.8001,-1.3);

\draw (10,4) .. controls (9,6) and (8,6) .. (6,2);

\draw[->] (8.8,5.3) -- (8.8001,5.3);

\draw (12,0) .. controls (10,-4) and (8,-4) .. (6,2);

\draw[->] (9.3,-2.8) -- (9.3001,-2.8);

\draw (12,4) .. controls (10,8) and (8,8) .. (6,2);

\draw[->] (9.3,6.8) -- (9.3001,6.8);


\draw[->] (6,2) -- (7,1);
\draw[-] (7,1) -- (8,0);

\draw[->] (6,2) -- (8,2);
\draw[-] (8,2) -- (10,2);

\draw[->] (6,2) -- (7,3);
\draw[-] (7,3) -- (8,4);


\end{tikzpicture}

\caption{A planar graph $\vec{B}_0$ of order 8 with $\chi_{p}([\vec{B}_0]) = 8$.}\label{fig push planarmax}
\end{figure}
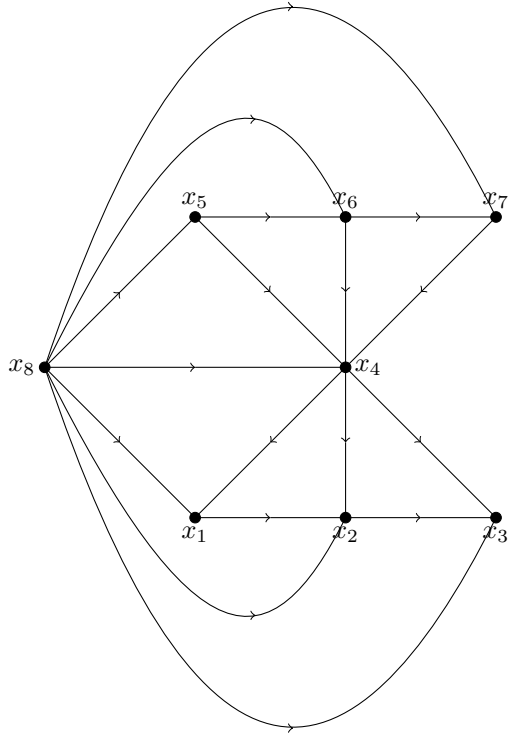

Due to the minimality of $\overrightarrow{H}$ there exists a planar push graph $[\overrightarrow{X}_0]$ with the following property: 
if $f$ is a homomorphism of a presentation ${\overrightarrow{X}}^{\text{\tiny{1}}}_0 \in [\overrightarrow{X}_0]$ to  $ \overrightarrow{H}$ then 
 for each arc $\overrightarrow{uv} \in A(\overrightarrow{H})$ there exists an arc 
$\overrightarrow{ab} \in A({\overrightarrow{X}}^{\text{\tiny{1}}}_0)$ such that $f(a) = u$ and $f(b) = v$.

Now we construct an oriented planar graph $\overrightarrow{X_1}$ by gluing a copy of the planar graph $\overrightarrow{B_0}$ (Fig.~\ref{fig push planarmax})
to each vertex of $\overrightarrow{X_0}$ by identifying the vertex with the vertex $x_8$ of $\overrightarrow{B_0}$.
Then we glue the gadget graph $\overrightarrow{Y}$  (Fig.~\ref{fig gadget})
to each vertex of $\overrightarrow{X_1}$ by identifying the vertex with the vertex $v$ of $\overrightarrow{Y}$ and obtain a new oriented graph $\overrightarrow{X_2}$.
Note that the gadget graph $\overrightarrow{Y}$ is planar and thus  $\overrightarrow{X_2}$ is also planar. 
After that we construct another oriented planar graph $\overrightarrow{X_3}$ by gluing a copy of the planar graph $\overrightarrow{B_0}$ 
to each arc of $\overrightarrow{X_2}$ by identifying the arc with the arc $\overrightarrow{x_8x_4}$ of $\overrightarrow{B_0}$.

Note that each pair of  non-adjacent vertices of the oriented planar graph $\overrightarrow{B_0}$ from Fig~\ref{fig push planarmax}
is part of a common $\overrightarrow{UC}_4$. 
So no two vertices of $\overrightarrow{B_0}$ can have the same homomorphic image 
and thus $\chi_{p}([\overrightarrow{B_0}]) = 8$. 
Moreover, $|A_{x_0,x_8}|, |D_{x_0,x_8}| \geq 3$ in  $\overrightarrow{B_0}$. Also as no two vertices of $\overrightarrow{B_0}$ can be identified,
if $[\overrightarrow{B_0}]$ admits a homomorphism $f$ to a graph $\overrightarrow{X}$ then  $|A_{f(x_0),f(x_8)}|, |D_{f(x_0),f(x_8)}| \geq 3$ in 
$\overrightarrow{X}$. 

Consider the push graph $[\overrightarrow{Y}]$ of the gadget graph $\overrightarrow{Y}$ (Fig.~\ref{fig gadget}). Here  the vertices $v, x$ and $y$ must have different image under any homomorphism as any pair of them are part an $\overrightarrow{UC}_4$. Also the vertices of the directed 5-cycle that is induced by the common neighbors of $v$ and $x$ must have distinct images under any homomorphism. Hence 
if $[\overrightarrow{Y}]$ admits a homomorphism $f$ to a graph $\overrightarrow{X}$ then  $M_{f(v),f(x)} \geq 4$ in 
$\overrightarrow{X}$. Similarly, we will have  $M_{f(v),f(y)} \geq 4$ in $\overrightarrow{X}$.

 Therefore,  each vertex in $\overrightarrow{H}$ has degree at least 7 and for each arc $\overrightarrow{ab}$ of $\overrightarrow{H}$ we have $m_{a,b} \geq 3$. Also for each vertex $a$ of $\overrightarrow{H}$ there are at least two vertices $c$ and $d$ of 
 $ \overrightarrow{H}$ with $M_{a,c}, M_{a,d} \geq 4$.

Theerefore,  for each vertex $u$ of  $ \overrightarrow{H}$ there exists a 
 vertex $v$ of $ \overrightarrow{H}$
with $m_{u,v} \geq 3$ and $M_{u,v} \geq 4$. That is,  $u$ and $v$ have at least 7 common neighbors. Thus $ \overrightarrow{H}$ has at least 
9 vertices and if $ \overrightarrow{H}$ has exactly 9 vertices  then it is a tournament.

\medskip

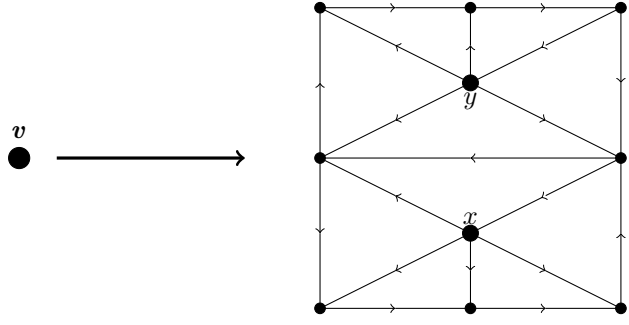
\begin{figure}

\centering
\begin{tikzpicture}

\filldraw [black] (-4,2) circle (4pt) {node[above]{}};

\node at (-4,2.35) {\textbf{\textit{v}}};

\draw[->][very thick] (-3.5,2) -- (-1,2);


\filldraw [black] (0,0) circle (2pt) {node[below]{}};
\filldraw [black] (2,0) circle (2pt) {node[below]{}};
\filldraw [black] (4,0) circle (2pt) {node[below]{}};

\filldraw [black] (0,4) circle (2pt) {node[below]{}};
\filldraw [black] (2,4) circle (2pt) {node[below]{}};
\filldraw [black] (4,4) circle (2pt) {node[below]{}};

\filldraw [black] (0,2) circle (2pt) {node[below]{}};
\filldraw [black] (4,2) circle (2pt) {node[below]{}};

\filldraw [black] (2,1) circle (3pt) {node[above]{$x$}};
\filldraw [black] (2,3) circle (3pt) {node[below]{$y$}};

\draw[->] (2,1) -- (1,.5);
\draw[-] (1,.5) -- (0,0);

\draw[->] (2,1) -- (2,.5);
\draw[-] (2,.5) -- (2,0);

\draw[->] (2,1) -- (3,.5);
\draw[-] (3,.5) -- (4,0);

\draw[->] (2,1) -- (1,1.5);
\draw[-] (1,1.5) -- (0,2);

\draw[-<] (2,1) -- (3,1.5);
\draw[-] (3,1.5) -- (4,2);


\draw[->] (2,3) -- (1,2.5);
\draw[-] (1,2.5) -- (0,2);

\draw[->] (2,3) -- (2,3.5);
\draw[-] (2,3.5) -- (2,4);

\draw[->] (2,3) -- (3,2.5);
\draw[-] (3,2.5) -- (4,2);

\draw[->] (2,3) -- (1,3.5);
\draw[-] (1,3.5) -- (0,4);

\draw[-<] (2,3) -- (3,3.5);
\draw[-] (3,3.5) -- (4,4);


\draw[->] (0,0) -- (1,0);
\draw[->] (1,0) -- (3,0);
\draw[-] (3,0) -- (4,0);

\draw[->] (4,0) -- (4,1);
\draw[-] (4,1) -- (4,2);

\draw[->] (4,2) -- (2,2);
\draw[-] (2,2) -- (0,2);

\draw[->] (0,2) -- (0,1);
\draw[-] (0,1) -- (0,0);

\draw[->] (0,2) -- (0,3);
\draw[-] (0,3) -- (0,4);

\draw[->] (0,4) -- (1,4);
\draw[->] (1,4) -- (3,4);
\draw[-] (3,4) -- (4,4);

\draw[->] (4,4) -- (4,3);
\draw[-] (4,3) -- (4,2);

\end{tikzpicture}

\caption{The gadget graph $\vec{Y}$. The thick arrow suggests that there are arcs from $v$ to every other vertices except $x$ and $y$.}\label{fig gadget}
\end{figure}

Now for the rest of the proof we will assume that $ \overrightarrow{H}$ is a tournament on 9 vertices and prove the theorem by contradicting 
this assumption.

First recall that there is a vertex $a$ of $ \overrightarrow{H}$ such that $N^+(a) = V(\overrightarrow{H}) \setminus \{a\}$. Moreover, there is a vertex $b$ of $ \overrightarrow{H}$ such that $|A_{a,b}| \geq  |D_{a,b}|$. This implies $|N^+(a) \cap N^+(b)| = 4$ and $|N^+(a) \cap N^-(b)| = 3$.
Now suppose that $N^+(a) \cap N^+(b) = A$ and $N^+(a) \cap N^+(b) = B$.

For any  vertex $x \in V(\overrightarrow{H}) \setminus \{a\}$ we have $m_{a,x} \geq 3$. 
Thus we must have $3 \leq d^+(x) \leq 4$ for all $x \in V(\overrightarrow{H}) \setminus \{a\}$.

\medskip

\textbf{\textit{Claim 3:}} Each vertex $x \in A$ must have exactly 1 in-neighbors in $B$. 

\medskip

\textit{Proof of the claim:} Suppose a vertex $x \in A$ does not have exactly 1 in-neighbors in $B$. Then one of the following cases hold:

\begin{itemize}
\item[$(i)$]  Suppose $x$ has no in-neighbors in $B$. That means $B \subseteq N^+(x)$. As  $3 \leq d^+(x) \leq 4$, there can be  at most 1 out-neighbor of $x$ in $A$. The set $A_{x,b}$ contains the vertex $a$, out-neighbors of $x$ in $A$ and in-neighbors of $x$ in $B$. Thus $|A_{x,b}| \leq 2$, a contradiction. 

\item[$(ii)$]  Suppose $x$ has exactly 2 in-neighbors in $B$.  As  $3 \leq d^+(x) \leq 4$, there must be at  least 2 out-neighbor of $x$ in $A$. 
Thus there is at most 1 in-neighbor of $x$ in $A$. The set $D_{x,b}$ contains in-neighbors of $x$ in $A$ and out-neighbors of $x$ in $B$. 
Thus $|D_{x,b}| \leq 2$, a contradiction.

\item[$(iii)$]  Suppose $x$ has exactly 3 in-neighbors in $B$.  As  $3 \leq d^+(x) \leq 4$, there must be at  least 3 out-neighbor of $x$ in $A$. 
Thus there is no in-neighbor of $x$ in $A$. The set $D_{x,b}$ contains in-neighbors of $x$ in $A$ and out-neighbors of $x$ in $B$. 
Thus $|D_{x,b}| = 0$, a contradiction. \hfill $ \bullet$
\end{itemize}

\medskip

As $|A| = 4$ and $|B| = 3$, by pigeonhole principal, there are two vertices $x,y \in A$ with a common in-neighbor in $B$.
That implies the other two vertices of $B$ are common out-neighbors of $x$ and $y$. Thus $A_{x,y} \supseteq B \cup \{a,b\}$. 
Therefore, $D_{x,y} \subseteq A \setminus \{x,y\}$. That implies $D_{x,y} \leq  |A \setminus \{x,y\}| \leq 2$, a contradiction.
\end{proof}

Note that, improving the upper bound will improve the long standing 
upper bound of oriented chromatic number of planar graphs. Indeed our result uses the proof of  the later. 
Whereas our lower bound proof is independent of the lower bound proof for oriented chromatic number of planar graphs by Marshall~\cite{marshall18}. 
Moreover, a lower bound of 9 for the push chromatic number of planar graphs can be achieved using Marshall's result 
while we provide a better lower bound of 10 for the same. Even though our lower bound does not imply any improvement of Marshall's lower bound of 18
for oriented chromatic number of planar graphs, it does imply the following corollary.

\begin{corollary}\label{cor 19?}
There exists no splitable oriented graph on 18 vertices to which every oriented planar graph admits a homomorphism to. 
\end{corollary}

A graph is called a \textit{core graph} if it does not admit a homomorphism to any of its proper subgraph~\cite{hell}.  The unique~\cite{hell} subgraph to which 
a graph admits a homomorphism to is called its \textit{core}. 
Marshall~\cite{marshall17} first established the lower bound of 17 for oriented chromatic number of planar graphs by showing that there exists no oriented graph on 16 vertices to which every planar graph admits a homomorphism to. 
For proving this first he showed that the Tromp graph~\cite{marshall17} $T_{16}$ on 16 vertices is the only graph to which every planar graph can admit a homomorphism to. Then he constructed an example of an oriented planar graph that does not admit a homomorphism to $T_{16}$.
After that Marshall~\cite{marshall18} extended his result to prove that the only oriented graph on 17 vertices to which all planar graphs can admit a homomorphism to is an oriented graph whose core is $T_{16}$. 
An easy but significant observation is that the family of Tromp graphs, in particular $T_{16}$, are splitable graphs. 
So if one can show that the only possible oriented core graph on 18 vertices to which every planar graph admits a homomorphism to is a splitable graph, then by our result the lower bound for oriented chromatic number can be improved to 19.

\medskip

\noindent \textbf{Question:} Is it possible to get rid of the word ``splitable'' from Corollary~\ref{cor 19?}?

\medskip
   
Now we will prove a tight bound for push chromatic number for the family of planar graphs with girth at least 8.

\begin{theorem}\label{th p8}
Let $\mathcal{P}_8$ be the family of all planar graphs with girth at least 8. Then $\chi_{p} ( \mathcal{P}_8) = 4$. 
\end{theorem}

\begin{proof}
If $\chi_{p}  \leq 3$, then we can prove that there exists a tournament on 3 vertices to which every 
planar push graph with girth at least 8 admits a homomorphism to by mimiking the proof of claim~1 of Theorem~\ref{th p3}. 
Now, upto push euivalence, there is only one tournament, the directed 3-cycle $\overrightarrow{C_3}$,  on three vertices. 
Therefore, given any planar graph $ \overrightarrow{H}$ with girth at least 8, there   must be a presentation 
of  $[\overrightarrow{H}]$ that admits a homomorphism to the directed 3-cycle $\overrightarrow{C_3}$ by Proposition~\ref{th any target}. 

Take the directed    9-cycle $\overrightarrow{C_9}$. Now construct the graph $\overrightarrow{H}$ by taking 
  $\overrightarrow{C_9}$ and a new vertex $v$ and then connecting each vertex of $\overrightarrow{C_9}$ to $v$ by two distinct paths of length 4 (one of them directed and the other with three forward arcs and one backward arc). 
Now consider the push graph $[\overrightarrow{H}]$. 

Notice that for any presentation $\overrightarrow{H} \in [\overrightarrow{H}]$ we will have  one 4-path, 
with either three forward arcs and one backward arc or with three backward arcs and one forward arc, connecting $v$ to each vertex of the 9-cycle. 

Observe  that the 4-path $x_0x_1...x_4$ with three forward arcs and one backward arc 
does not admit a homomorphism with $x_0$ and $x_4$ mapped to the same vertex of $\overrightarrow{C_3}$.
Now let $f$ be a homomorphism of $\overrightarrow{H}$ to $\overrightarrow{C_3}$. Then, because of the above observation, 
$f(v) \neq f(u)$ for every vertex $u$ from the 9-cycle. But we know that the 9-cycle has push chromatic 
number equal to 3. That means $f$ must be onto on the vertices of $\overrightarrow{C_3}$ when restricted to 
the 9-cycle. Hence $f(v) \notin V(\overrightarrow{C_3})$. This is a contradiction.
Hence we have the lower bound.  

\medskip

For proving the upper bound it is enough to show that every push $[\overrightarrow{G}]$ with  
maximum average degree less than $8/3$   admits a homomorphism to the Paley plus graph $\overrightarrow{P}_3^+$ due to Borodin,  Kostochka, 
 Ne\v{s}et\v{r}il,  Raspaud and Sopena~\cite{mad}. 
We will use the discharging method for our proof.

We first provide a (small) set of \textit{forbidden confgurations}, that
is a set of graphs that a minimal counterexample $[\overrightarrow{H}]$ to our claim cannot contain
as subgraphs. We will then assume that every vertex $v$ in $[\overrightarrow{H}]$ is valued by its degree
$deg(v)$ and define a \textit{discharging procedure} which specifies some transfer of values
among the vertices in $[\overrightarrow{H}]$, keeping the sum of all the values constant. We will then get
a contradiction by considering the \textit{modernized degree} $deg^*(v)$ of every vertex $v$, that
is the value obtained by $v$ owing to the discharging procedure.

\medskip

\textbf{\textit{Drawing conventions:}} In all the figures depicting forbidden configurations, we will
draw vertices with prescribed degrees as `square vertices' and vertices with unbounded
degree as `circular vertices'. All the neighbors of square vertices are drawn. Unless otherwise
specified, two or more circular vertices may coincide in a single vertex, provided
that they do not share a common square neighbor.

\textbf{\textit{Observation 1:}} It is easy to check that 
$N^{++}(\{v\}) \cup N^{--}(\{v\}) = V(\overrightarrow{P}_3^+) \setminus \{v\}$
and  $N^{+-}(\{v\}) \cup N^{-+}(\{v\}) = V(\overrightarrow{P}_3^+)$ for all $v \in V(\overrightarrow{P}_3^+)$.

First assume that $[\overrightarrow{H}]$ is a mimimal (with respect to the number of vertices) push graph
with maximum average degree less than $8/3$ that does not admit a homomorphism to
 $\overrightarrow{P}_3^+$.

\begin{figure}

\centering
\begin{tikzpicture}

\filldraw [black] (-4,5) circle (2pt) {node[below]{}};
\filldraw [black] (-3,6) circle (2pt) {node[below]{}};
\filldraw [black] (-2,5) circle (2pt) {node[below]{}};
\filldraw [black] (-3,7) circle (2pt) {node[below]{}};

\draw[->] (-4,5) -- (-3,5);
\draw[-] (-3,5) -- (-2,5);

\draw[->] (-2,5) -- (-2.5,5.5);
\draw[-] (-2.5,5.5) -- (-3,6);

\draw[->] (-3,6) -- (-3.5,5.5);
\draw[-] (-3.5,5.5) -- (-4,5);

\draw[->] (-3,7) -- (-2.5,6);
\draw[-] (-2.5,6) -- (-2,5);

\draw[->] (-3,7) -- (-3,6.5);
\draw[-] (-3,6.5) -- (-3,6);

\draw[->] (-3,7) -- (-3.5,6);
\draw[-] (-3.5,6) -- (-4,5);

\node at (-3,4) {$(a)$};

\path (0,6) edge (1,6);

\filldraw [black] (0,6) circle (2pt) {node[below]{}};

\filldraw [black] ([xshift=-2pt,yshift=-2pt]1,6) rectangle ++(4pt,4pt) {node[above]{}};

\node at (.5,5) {$(i)$};

\filldraw [black] (2,6) circle (2pt) {node[below]{}};
\filldraw [black] (5,6) circle (2pt) {node[below]{}};

\filldraw [black] ([xshift=-2pt,yshift=-2pt]4,6) rectangle ++(4pt,4pt) {node[above]{}};
\filldraw [black] ([xshift=-2pt,yshift=-2pt]3,6) rectangle ++(4pt,4pt) {node[above]{}};

\draw[-] (2,6) -- (5,6);

\node at (3.5,5) {$(ii)$};

\filldraw [black] (6,7) circle (2pt) {node[above]{$u_1$}};
\filldraw [black] (6,5) circle (2pt) {node[above]{$u_2$}};
\filldraw [black] (9,6) circle (2pt) {node[right]{$u_3$}};

\filldraw [black] ([xshift=-2pt,yshift=-2pt]7,6.5) rectangle ++(4pt,4pt) {node[above]{$v_1$}};
\filldraw [black] ([xshift=-2pt,yshift=-2pt]7,5.5) rectangle ++(4pt,4pt) {node[above]{$v_2$}};
\filldraw [black] ([xshift=-2pt,yshift=-2pt]8,6) rectangle ++(4pt,4pt) {node[above]{$v_3$}};

\draw[-] (8,6) -- (9,6);

\draw[-] (8,6) -- (6,7);

\draw[-] (8,6) -- (6,5);

\node at (7.5,5) {$(iii)$};

\node at (4,4) {$(b)$};

\end{tikzpicture}

\caption{$(a)$ The Paley plus graph $P^+_3$. $(b)$ The forbidden configurations for Theorem~\ref{th p8}.}~\label{figure orientable girth8}

\end{figure}
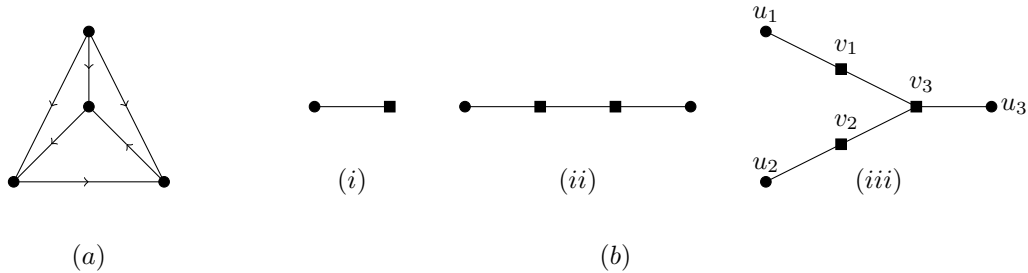

First we will show that  $[\overrightarrow{H}]$  does not contain any of the configuration depicted in 
Fig.~\ref{figure orientable girth8}.

\begin{itemize}

\item[(i)] Obvious since every vertex of $\overrightarrow{P}_3^+$  has degree at least one.

\item[(ii)] Directly follows from Observation~1.

\item[(iii)] Consider the push graph $[\overrightarrow{H'}]$ obtained by deleting all the square vertex of degree 3 from $[\overrightarrow{H}]$. Therefore, there exists a presentation 
${\overrightarrow{H'}}^{\text{\tiny{1}}} \in [\overrightarrow{H'}]$ such that ${\overrightarrow{H'}}^{\text{\tiny{1}}}$ admits a homomorphism $f'$ to $\overrightarrow{P}_7$.

Now choose a vertex $x \in V(\overrightarrow{P}_3^+) \setminus \{f'(u_1),f'(u_2),f'(u_3)\}$. 
Suppose that $x \in N^\alpha_{\overrightarrow{P}_3^+}(f'(u_3))$. 

Now consider the presentation ${\overrightarrow{H}}^{\text{\tiny{1}}} \in [\overrightarrow{H}]$ that contains ${\overrightarrow{H'}}^{\text{\tiny{1}}}$ as a subgraph
and is such that $v_3 \in N^\alpha_{{\overrightarrow{H}}^{\text{\tiny{1}}}}(u_3)$ (such a presentation is possible to obtain by 
pushing $v_3$ if needed). 

Now we can extend $f'$ to a homomorphism 
of $f$ of ${\overrightarrow{H}}^{\text{\tiny{1}}}$ to  $\overrightarrow{P}_3^+$ (by pushing the vertices $v_1$ and $v_2$ if needed) with $f(v_1) = x$ using 
 Observation~1.

\end{itemize}

We now use the following discharging procedure: each vertex of degree at least
3 gives $1/3$ 
to each of its neighbors with degree $2$.

Let us check that the modernized degree $deg^*(v)$ of each vertex $v$ is at least $8/3$ which
contradicts the assumption $mad(H)< 8/3$. We consider the possible cases for the old
degree $deg(v)$ of $v$:

\begin{itemize}

\item[(i)]  $deg(v) = 1$: there is no such vertex in $[\overrightarrow{H}]$ by (i).

\item[(ii)]$deg(v) = 2$: by (ii), both its neighbors have degree at least 3. Therefore, it receives 
exactly
$2 \times 1/3 = 2/3$, and thus $deg^*(v) = 2+2/3 = 8/3$.

\item[(iii)] $deg(v) = 3$: by  (iii),  gives away at most $1/3$. Therefore, we have 
 $deg^*(v) = 3 - 1/3 = 8/3$.

 \item[(iv)] $deg(v) = k \geq 4$:  it gives away at most $k \times 1/3 = k/2$. Therefore, we have 
 $deg^*(v) \geq  k-k/3 = 2k/3 \geq  8/3$.

\end{itemize}

 Therefore, every vertex of $[\overrightarrow{H}]$ gets a modernized degree at least $8/3$. Hence, every push graph with maximum average degree less than $8/3$ admits a homomorphism to $\overrightarrow{P}_3^+$. 

\end{proof}

\bibliographystyle{abbrv}
\bibliography{POreferences}

\end{document}